\documentclass{birkjour}
 \newtheorem{thm}{Theorem}[section]

 \theoremstyle{definition}
 
 \theoremstyle{remark}
 \newtheorem{rem}[thm]{Remark}
 
 \numberwithin{equation}{section}
\usepackage{url}
\newtheorem{defi}[thm]{Definition}
\usepackage{subfig}
\usepackage{amsmath,amssymb}
\usepackage{amsfonts}
 \usepackage{caption}
\begin{document}

\title[Comparative prediction of confirmed cases with COVID-19 pandemic.]{Comparative prediction of confirmed cases with \\ COVID-19 pandemic by machine learning,   \\    deterministic and stochastic   SIR models}%
\author[Ndiaye et al.]{Babacar Mbaye Ndiaye, Lena Tendeng, Diaraf Seck}
\address{Laboratory of Mathematics of Decision and Numerical Analysis\br
University of Cheikh Anta Diop\br
BP 45087, 10700. Dakar, Senegal}

\email{babacarm.ndiaye@ucad.edu.sn, lenatendeng@yahoo.fr,\\ diaraf.seck@ucad.edu.sn}

\thanks{This work was completed with the support of the NLAGA project}

\date{April 20, 2020}
\begin{abstract}
In this paper, we propose a machine learning technics and SIR models (deterministic and stochastic cases) with numerical approximations to predict the number of cases infected  with the COVID-19, for both in few days and the following three weeks.  
Like in \cite{ndiaye} and based on the public data from \cite{datahub}, we estimate parameters and make predictions to help on how to  find concrete actions to control the situation. 
Under optimistic estimation, the pandemic in some countries will end soon, while for most of the countries in the world, the hit of anti-pandemic will be no later than the beginning of May.
\end{abstract}

\maketitle

\section{Introduction}\label{intro}
The deterministic SIR models are usually investigated through ordinary differential equations for prediction \cite{howard}. It can also be viewed in a stochastic framework,  which is more realistic but also more complicated to analyze.\\
In population dynamics, the deterministic models are developed  with success in many situations. In these models appear some parameters  and in concrete applications, its  estimated values play  a crucial role in the prediction of the  studied system  and even in the decision-making policies. Usually one considers that they  are  deterministic, but many times due to errors in measurements, variability in the populations, and other factors that introduce uncertainties, one can think of  parameters as random variables. 
To consider these aspects, it is necessary to get some skills in probability theory, statistics, and differential equations.  And the stochastic differential equation theory invites itself and fortunately there are many excellent books in this topics. For more details see for instance \cite{LCE,Ok,IM}.\\
And with tools developed in this topic, interdisciplinary areas such as mathematical biology, biostatistics, and bio-engineering have become possible with rapid growth. There are interesting works done and the reader can see for example \cite{A1,A2,LS} and references therein.\\
Our aim in this work is one hand, to propose a stochastic model to analyze the pandemic of COVID-19. And on the other hand, we would like to deepen the numerical analysis  of  such phenomena in situations where the  settings may be  random.\\
Many theoretical studies of the evolution of infectious diseases of  the COVID-19 are recently proposed in \cite{ndiaye,seydi, hiroshi,Kamalich,steven}.  
\noindent In the simple SIR model, the total population for each country \cite{pyramid} is assumed to be constant and divided into three classes (susceptible, infected, and recovered).  
In the numerical simulations, we start with the deterministic case, following by the new proposed stochastic model in section \ref{model}, then 3 others SIR models and machine learning for forecasting where algorithms can include artificial neural networks, deep learning, association rules, decision trees, reinforcement learning and bayesian networks \cite{arkes,litsa}.\\
\noindent First, we collect carefully the pandemic data from \cite{datahub},  
e.g. \url{https://www.tableau.com/covid-19-coronavirus-data-resources},  from January 21, 2020 to April 19, 2020. After exploratory data analysis,  we propose six(6) technics,  a simple SIR model, a stochastic SIR with Brownian motion, SIR with Deaths, SIR with Fatal, SIR Exposed and Waiting cases with Fatal, and machine learning tools, to analyze the coronavirus pandemic in the worldwide. A special study is done for Senegal.

\noindent The paper is organized as follows. In section \ref{model}, we present a stochastic SIR model with the existence, uniqueness and some qualitative results. In section \ref{numsim}, we  present approximation methods to estimate different parameters involving in the  SIR models. It is followed by various numerical tests for comparative prediction. Finally, in section \ref{ccl}, we present conclusions and perspectives.

\section{Modeling, Existence, Uniqueness and  Properties}\label{model}
\subsection{Stochastic Model}
The stochastic aspect would be very interesting due to the lack of data, especially in the case of Senegal with the expansion of COVID 19 pandemic. Some infected people do not develop the disease and spread it (the case of some children, as it is suspected) not to mention the unmonitored asymptomatic cases  that  are the cause of community transmission. There is also the probability of touching infected objects:  too many random factors in the transmission of the disease. Let us propose  among many possibilities the following stochastic SIR model:
 
 \begin{equation}\label{stoeq1}
\left\{ 
\begin{array}{ll}
\frac{dS}{dt}=-\beta IS + \sigma_1 \xi_1\\[3mm]
\frac{dI}{dt}= \beta IS-\gamma I + \sigma_2 \xi_2\\[3mm]
\frac{dR}{dt}=\gamma I
\end{array}\right.
\end{equation}
\begin{itemize}
\item $S$ is the number of individuals susceptible to be infected at time $t$.
\item $I$ is the number of both asymtomatic and symptomatic infected individuals at time $t$.
\item $R$ is the number of recovered persons at time $t$.
\item The parameters $\beta$ and $\gamma$ are  respectively the transmission rate through exposure of the disease and the rate of recovering.
\item $\sigma_1$ and $\sigma_2$ are  diffusion coefficients that are interpreted as  volatility rates. They  may  depend on the time $t$, the suspects and the infected. $\xi_1$ and  $\xi_2$ are  white noises.
\end{itemize}
One interesting case that  we endeavor  to look at  in this work, is when  the above system is written as  follows:
\begin{equation}\label{stoeq2}
\left\{ 
\begin{array}{ll}
dS=-\beta IS dt   - \sigma_1 I S dW_1\\[3mm]
dI= \beta IS dt -\gamma I + \sigma_2 I S dW_2\\[3mm]
\frac{dR}{dt}=\gamma I
\end{array}\right.
\end{equation}
 $W_1$ and $W_2$ are classical Brownian motions;  $\sigma_1$ and $\sigma_2$ are positive constants.
 \begin{rem}
 Let us note that  instead  the term $\sigma_i IS dW_i, i= 1, 2$  one could propose   $\sigma_i \sqrt{IS} dW_i, i= 1, 2.$
 \end{rem}
\noindent As in the deterministic case,  let us consider  a sample of population. We denote it by $N$ satisfying the relation
 \begin{eqnarray}
 N= S(t)+ I(t)+ R(t)
 \end{eqnarray}
 The balance property of the above equation  implies the  following constraint
 \begin{eqnarray}\label{bal}
 -\sigma_1 dW_1+  \sigma_2 dW_2= 0.
 \end{eqnarray}
 With this constraint, we shall need in the simulation of the  stochastic case  to have an estimation of  one of $\sigma_1, \sigma_2.$\\
  Before proceeding further,  we are  going  to  present some classical results  for the  stochastic  differential equation model,  such as existence and uniqueness results, the stopping notion, and its properties. 

\subsection{Existence and Uniqueness}
This section is  started by a few  useful reminds in probability  theory. For more details, see for instance \cite{LCE,Ok}. 
\begin{defi}
Let $\Omega$ be a set.\\
A $\sigma-$algebra is a collection $\mathcal U$ of subsets of $\Omega$ with the following properties
\begin{itemize}
\item $\emptyset, \Omega \in \mathcal U;$
\item if $A\in \mathcal U,$ then $A^c:= \Omega \backslash A \in  \mathcal U;$
\item if $A_1,  A_2,..., \in \mathcal U$ then  $\displaystyle \bigcup_{i=1}^{\infty}A_i,  \bigcap_{i=1}^{\infty}A_i \in   \mathcal U.$
\end{itemize}
\end{defi}
\noindent Let $ W (.) $ be a $1-$dimensional Brownian motion defined on some probability space $(\Omega, \mathcal U , \mathbb{P} ).$
\begin{defi}
The $\sigma-$algebra $\mathcal W (t):= \mathcal U (W (s)/ 0\leq s\leq t)$ is called the history of Brownian motion up to and including $t.$\\
The $\sigma-$algebra $\mathcal W^{+}(t):= \mathcal U (W (s)/ s\geq t)$ is called the future of Brownian motion beyond time  $t.$
\end{defi}
\begin{defi}
A family $\mathcal F (.)$ of $\sigma-$ algebras  included in $\mathcal U$ is called  non anticipating with respect to $W(.)$ if
\begin{enumerate}
\item $\mathcal F (t) \supseteq \mathcal F (s)$ for all $t\geq s \geq 0;$
\item $\mathcal F (t) \supseteq \mathcal W (t)$ for all $t \geq 0;$
\item $\mathcal F (t) $ is independent of $\mathcal W^{+}(t)$ for all $t \geq 0;$
\end{enumerate}
It is referred to $ \mathcal F (t)$ as a filtration.
\end{defi}
\noindent Let us state  a general  a Cauchy Lipschitz Theorem  version and for the details (\ref{stoeq2}), see for example \cite{LCE}.
\begin{thm}\label{GCL}
Suppose that the two functions $\textbf{b}:\mathbb R^n\times[0,T]\longrightarrow \mathbb R^n$ and $\textbf{B}:\mathbb R^n\times[0,T]\longrightarrow \mathbb M^{m\times n}$ are continuous and satisfy the following conditions:
\begin{itemize}
\item[(a)] $\mid \textbf{b}(x,t)-\textbf{b}(\hat x,t)\mid \leq L\mid x-\hat x\mid $  \;and\; $\mid \textbf{B}(x,t)-\textbf{B}(\hat x,t)\mid \leq L\mid x-\hat x\mid $    $\forall \;0\leq t\leq T  \; and\; x,\hat x \in \mathbb R^n$.
\item[(b)] $\mid \textbf{b}(x,t)\mid \leq L(1+\mid x\mid) $ and $\mid \textbf{B}(x,t)\mid \leq L(1+\mid x\mid) $ $\forall \;0\leq t\leq T  \;  , x \in \mathbb R^n$ 
for some constant L.\\
Let $\textbf{X}_0$ be any $\mathbb R^n$-valued random variable such that
\item[(c)] $E(\mid \textbf{X}_0\mid^2)<\infty$, and
\item[(d)] $\textbf{X}_0$ is independent of  $\mathcal W^+(0)$,
where $\textbf{W}(.)$ is a given m-dimensional Brownian motion.
\end{itemize}
Then there exists a unique solution $\textbf{X}\in \mathbb L^2_n(0,T)$ of the stochastic differential equation:
\begin{equation}\label{stoeq}
\left\{ 
\begin{array}{ll}
d\textbf{X}=\textbf{b}(\textbf{X},t){dt} + \textbf{B}(\textbf{X},t) d\textbf{W}  \ \  \ \ (0\leq t\leq T) \\
\textbf{X}(0)=\textbf{X}_0
\end{array}\right.
\end{equation}
\end{thm}
\noindent The above theorem is merely adapted in our study case.
\begin{thm}
Let $\textbf{X}_0 = (S_0^*, I_0^*, R_0^*)$ be any $\mathbb R^3$-valued random variable such that:
\begin{itemize}
\item[(i)] $E(\mid \textbf{X}_0\mid^2)<\infty$, and
\item[(ii)] $\textbf{X}_0$ is independent of  $\mathcal W^+(0)$,
\end{itemize}
There exists a unique solution of the following stochastic differential equation:
\begin{equation}\label{stoeq}
\left\{ 
\begin{array}{ll} 
dS=-\beta IS dt - \sigma_1 IS dW_1\\[3mm]
dI= (\beta IS-\gamma I )dt+ \sigma_2  ISdW_2\\[3mm]
\frac{d{ R}}{dt}=\gamma I\\ [3mm]
S(0)=S_0^*, \ I(0)=I_0^*, \ R(0)=R_0^*
\end{array}\right.
\end{equation}
\end{thm}

\begin{proof}
The proof is simple. It suffices only to verify if the hypotheses of  Theorem \ref{GCL} are satisfied. 
$$ \left(\begin{array}{ccc}
                                     dS \\
                                     
                                     dI \\
                                     
                                     dR\\
                                   \end{array}
                                 \right)   =                                  
                             \left(\begin{array}{ccc}
                                     -\beta I & 0& 0\\
                                     \beta I & -\gamma&0\\
                                     0        & \gamma & 0\\
                              
                                   \end{array}
                             \right)       
                             \left(\begin{array}{ccc}
                                     S \\
                                     
                                     I \\
                                     
                                     R\\
                                   \end{array}
                                 \right)  dt                                  
     +     \left(\begin{array}{ccc}
                                     -\sigma_1IS \\
               \sigma_2IS \\
                                     
                                     0\\
                                   \end{array}\right) d\textbf{W}                                
     $$
Our system fits well to the one considered in the above general theorem.
with:                                
 $\textbf{X}= \left(\begin{array}{ccc}
                                     S \\
                                   
                                     I \\
                                     
                                     R\\
                                   \end{array}
                                 \right)  $,      \ $\textbf{W} =\left(\begin{array}{ccc}
                                     W_1 \\
                                     
                                       W_2 \\ 
                                     
                                     0
                                   \end{array}
                                 \right)$,  \ $ \textbf{b}(\textbf{X})=  
                            \left(\begin{array}{ccc}
                                     -\beta I & 0& 0\\
                                     \beta I & -\gamma&0\\
                                     0        & \gamma & 0\\
                              
                                   \end{array}
                             \right)  $     \; and    $\textbf{B}(\textbf{X})= \left(\begin{array}{ccc}
                                   -  \sigma_1 IS \\
                               
                                     \sigma_2 IS \\
                                  
                                     0\\
                                   \end{array}
                                 \right)  $. 
                                 
\noindent By chosen      $\hat X= \left(\begin{array}{ccc}
                                     \hat S \\
                                     
                                     \hat I \\
                                     
                                    \hat R\\
                                   \end{array}
                                 \right)   $, we have:  
                                                  
  $ \textbf{b}(X)-\textbf{b}(\hat X)=  
                            \left(\begin{array}{ccc}
                                     -\beta (I-\hat I) & 0& 0\\
                                     \beta (I-\hat I) & 0&0\\
                                     0        & 0 & 0\\
                              
                                   \end{array}
                             \right)   \Longrightarrow \mid\mid  b(X)-b(\hat X)\mid\mid_\infty=\beta|I-\hat I| $                                                              
                     
 $X-\hat X= \left(\begin{array}{ccc}
                                    S- \hat S \\
                                     
                                  I-   \hat I \\
                                     
                                  R-  \hat R\\
                                   \end{array}
                                 \right) \Longrightarrow \mid\mid  X-\hat X\mid\mid_\infty=\max(|S-\hat S|, |I-\hat I| ,|R-\hat R|)  $

 $$|I-\hat I|\leq \max(|S-\hat S|,|I-\hat I|, |R-\hat R|)$$
 $  \Longrightarrow \mid\mid  b(X)-b(\hat X)\mid\mid_\infty=\beta|I-\hat I| \leq \beta \mid\mid  X-\hat X\mid\mid_\infty$. We have $\beta(1+\mid\mid X\mid\mid) \geq 0$ 
$$\mid\mid b(X)\mid\mid_\infty=\max(\beta I, \beta I+\gamma,\gamma)=\beta I+\gamma  \quad \mbox{and} \quad \mid\mid X\mid\mid_\infty=\max(S,I,R)$$
$\beta I\leq\beta \max(S,I,R)=\beta\mid\mid X\mid\mid $ \,or $\gamma<\beta  \Longrightarrow$
$$  \mid\mid b(X)\mid\mid_\infty=\beta I+\gamma\leq \beta+\beta\mid\mid X\mid\mid=\beta(1+\mid\mid X\mid\mid)$$ 
$$\displaystyle \mid\mid B\mid\mid=\sup_{\|X\|\leq 1}\frac{|B(X)|}{\|X\|}\leq C_0 +C_0  \max(S,I,R)= C_0 (1+\mid\mid X\mid\mid)$$
where $C_0= \max\{\sigma_1, \sigma_2\} $. 
We see that it suffices to take $L= \max\{\beta, C_0 \}$
\end{proof}

\subsection{Stopping time}\label{stopping}
Let $(\Omega, \mathcal U, \mathbb{P})$ be a probability space and $\mathcal F(.)$ a filtration of $\sigma-$algebras.
\begin{defi}
A random variable $\tau: \Omega \Longrightarrow [0, +\infty]$ is called  a stopping time  with respect to $\mathcal F (.)$ if  the set 
\begin{eqnarray*}
\{\tau \leq t\}, \, \, \mbox{for all}\, \, t \geq 0.
\end{eqnarray*}
\end{defi}
\noindent This means  that the set of all $\omega \in \Omega$ such that $\tau (\omega)\leq t $ is an $\mathcal F(t)-$ measurable set.
For our  model, it could be interesting to study the stopping time. For instance, the situation in which efforts are done to contain or eradicate the pandemic. And an interesting  one  is  can we find a finite stopping time when  the  susceptible infected and recovered or removed   variables   $S(t),$ $I(t)$ and  $R(t)$  respectively   are  less  than thresholds?\\
Let us state a general theorem  where the stopping time does exist but it may be  taken $+\infty.$

\begin{thm}
Let $V$ be either a nonempty closed subset or a nonempty open subset of $\mathbb{R}^3$. Then 
\begin{eqnarray}
\tau:= \inf\{t\geq 0 / (S(t), I(t), R (t)) \in V\}
\end{eqnarray}
is a stopping time
\end{thm} 
\noindent Let us point out that  the stopping time has interesting properties as a remark. 
Let $\tau_1$ and $\tau_2$ be stopping times with respect to $\mathcal F (.),$ then:
\begin{itemize}
 \item $\{\tau < t\} \in \mathcal F(t),$ and so $ \{\tau = t\} \in \mathcal F(t)$ for all times $t\geq 0.$
\item  $\tau_1\wedge \tau_2:= \min (\tau_1, \tau_2), \tau_1 \vee \tau_2: =\max (\tau_1, \tau_2) $ are stopping times
\end{itemize}
In the numerical simulations, we shall wonder if it is possible to find a finite stopping time by considering 
 $V$ a closed  subset of $\mathbb{R}^3$ in the form of $[0, S^*] \times [0, I^*]\times[0, R^*].$
 
\section{Parameter estimation and numerical simulations}\label{numsim}
In this section, we present the simulations of the simple SIR, stochastic SIR, SIR with Deaths, SIR with Fatal, SIR Exposed and Waiting cases with Fatal and the machine learning technic for forecasting of the pandemic. 
The numerical tests were performed by using the Python with the Panda library 
\cite{python}. The numerical experiments were executed on a computer with the following characteristics: intel(R) Core-i7 CPU 2.60GHz, 24.0Gb of RAM, under the UNIX system.

\subsection{Exploratory data analysis}\label{eda}
As stated in the introduction, the simulations are carried out from  data in \cite{datahub}, from January 21 2020, to April 19, 2020. 
We first analyze and make some data preprocessing before simulations. It is a good practice to know data types, along with finding whether columns contain null values or not.\\
We get various summary statistics, by giving the count (number of observations), mean, standard deviation, minimum and maximum values, and the quantiles of the data (see tables \ref{futurconfcases} and \ref{totdat}). 
\begin{table}[h!] 
\begin{center}
\begin{tabular}{|c|c|c|c| } 
 \hline
        {\bf Values }        &        {\bf Confirmed} &  {\bf Deaths } & {\bf  Recovered} \\
           \hline
count  & 16729 &  16729 & 16729 \\
           \hline
mean    & 	2532.4691 & 143.0855 &  	624.3309\\
           \hline
std &  	13183.2366  & 	1140.9791 &	4818.1358 \\
           \hline
min   &	0 & 	0   &	0\\
           \hline
25\%   & 	8   &	0  & 	0 \\
           \hline
50\% &	85   &	1  & 	1  \\
           \hline
75\% 	& 576 & 	6  & 	52 \\
           \hline
max  &	247815  & 	23660   &	88000 \\
\hline
\end{tabular}
\end{center}
\caption{Worldwide summary statistics (per day) until April 19th, 2020.}\label{futurconfcases}
\end{table}
\par\vspace{-.75cm}
\begin{table}[h!] 
\begin{center}
\begin{tabular}{|c|c|c|c| c| } 
 \hline
{\bf Date} & {\bf Confirmed}   &{\bf Infected}  & {\bf Deaths} &{\bf Recovered}  \\	
\hline
2020-04-17 	&  2240191 	  &  1518026   &	153822   &  	568343 \\
\hline
2020-04-18  & 	2317759 	&1565930 &	159510 &	592319 	\\
\hline
2020-04-19 &	2401379 &	1612432 	&165044 &	623903 \\
\hline
\end{tabular}
\end{center}
\caption{Worldwide total data until April 19th, 2020.}\label{totdat}
\end{table}

\noindent The worldwide cumulative of confirmed, deaths and recovered cases are illustrated in Figure \ref{ww_cidr}.
\begin{figure}[h!]
	\centering
	\includegraphics[width=1.\linewidth]{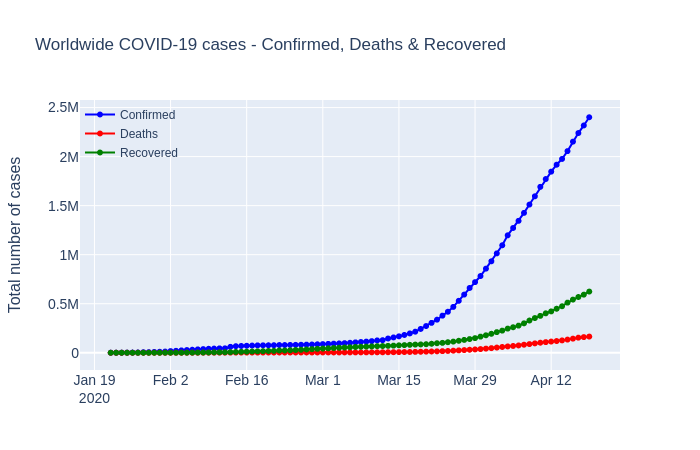}
		\par\vspace{-1.cm}
	\caption{Worldwide - confirmed, deaths and recovered} \label{ww_cidr}
\end{figure}

\subsection{Parameter estimation}
The identification of a real dynamical system (called object) is to characterize another system (called model), starting from the experimental knowledge of the inputs and outputs to obtain  an identity of behavior. In practice, the purpose of identification is generally to determine the conducted model, which can be used to simulate, control, or regulate a process. This model can be physical (in the sense of analog or digital simulator and reduced model), or  abstract (mathematical model, i.e. system of algebraic or differential equations (ODE or PDE)). \\
This subsection is started by the estimation of the parameters $\beta, \gamma$  in the deterministic SIR model by the standard least square method. 

\subsubsection{Deterministic case}
\begin{equation}\label{det}
\left\{ 
\begin{array}{ll}
\frac{dS}{dt}=-\beta IS\\[3mm]
\frac{dI}{dt}= \beta IS-\gamma I\\[3mm]
\frac{dR}{dt}=\gamma I
\end{array}\right.
\end{equation} 
\noindent Let us consider  a time interval $[0, T]$ and subdivide it as follows:
$\forall  i\in\{0,..., n-1\}$, $t_{i+1}- t_i =1$,  $t_0= 0$, $t_{n-1}= T$. 
Let: $\beta (t_i)= \beta_i, \gamma (t_i)= \gamma_i, S (t_i)= S_i$, and $I(t_i)= I_i$.\\
Discretizing $(\ref{det}),$ we get the following system of $2n$ unkown ($\forall i= 0,..., n-1$):
\begin{equation}\label{det2}
\left\{ 
\begin{array}{ll}
S_{i+1}- S_i =-\beta_i I_i S_i\\[3mm]
I_{i+1}- I_i= \beta I_i S_i-\gamma_i I_i\\[3mm]
R_{i+1}- R_i=\gamma_i I_i
\end{array}\right.
\end{equation} 
\noindent Let us set: $$Y= (S_1,,..., S_n, I_1,..., I_n) \in \mathbb{R}^{2n}\quad \mbox{and} \quad F (\beta_0,..., \beta_{n-1}, \gamma_0,..., \gamma_{n-1})\in \mathbb{R}^{2n} $$ the vector coming from the right hand side of the above two first equations of $(\ref{det2}).$
 $$P:= (\beta_0, \beta_1,..., \beta_{n-1}, \gamma_0, \gamma_1,..., \gamma_{n-1})\in [0, 1]^{2n}. $$
We try to minimize:
 $$\mathcal E (P):= \|Y-F(P)\|_2^2$$ on  $ [0, 1]^{2n} $, where $\|.\|_2$ is the Euclidean norm in  $ \mathbb{R}^{2n}.$\\
The differentiability  of $\mathcal E$ and  its convex structure ensure existence of a minimizer that we note by $P_{opt} = (\beta_0^*, \beta_1^*,..., \beta_{n-1}^*, \gamma_0^*, \gamma_1^*,..., \gamma_{n-1}^*).$\\
 And the approximated parameters proposed are:
 \begin{eqnarray*}
 \beta_{app}= \frac{1}{n}\sum_{i=0}^{n-1} \beta_i^* \qquad \mbox{and}\qquad
 \gamma_{app}= \frac{1}{n}\sum_{i=0}^{n-1} \gamma_i^*
 \end{eqnarray*}
We have:
\begin{itemize}
\item $S$ : Susceptible (= $N$ - Confirmed)
\item $I$ : Infected (= Confirmed - Recovered - Deaths)
\item $R$ : Recovered (= Recovered + Deaths)
\item $S+I+R=N$, where $N$ is the total population that can be obtained in \cite{pyramid}.
\end{itemize}
\noindent The basic reproduction number (also called basic reproduction ratio) is defined as  is $R_0 = \beta \gamma^{-1}$. 
This ratio is derived as the expected number of new infections (these new infections are sometimes called secondary infections) from a single infection in a population where all subjects are susceptible.
For the model, we also have:
$$ S+I \overset{\beta}{\longrightarrow} 2I \qquad  \mbox{and}\qquad  I \overset{\gamma}{\longrightarrow} R $$ 
with: $\beta$= effective contact rate [1/day], \ 
$\gamma$= recovery(+mortality) rate [1/day].
\begin{rem}
An easy way to compute $\gamma_i$, is to use the equation $R_{i+1}- R_i=\gamma_i I_i$, and set $y_i=R_{i+1}- R_i$, $x_i=I_i$\quad ($\forall i=1,...,n-1$). Then, we obtain $y_i=f(x_i)$. \\
Finally, as  $\gamma_i$ are bounded $\forall i=1,...,n$, we can call the {\tt curve\_fit} procedure. The {\tt scipy.optimize.curve\_fit} use non-linear least squares to fit a function, $f$, to data, assuming {\tt ydata = f(xdata, *params) + $\epsilon$}.\\
The return value {\tt popt} contains the best-fit values of the parameters. The return value {\tt pcov} contains the covariance (error) matrix for the fit parameters. From them, we can determine the standard deviations of the parameters. We can also determine the correlation between the fit parameters.
\end{rem}
\noindent To estimate $\beta$, we use the same procedure with the first equation $y_i = S_i-S_{i+1} =\beta_i I_i S_i=\beta_i x_i, \quad (\forall i=1,...,n-1)$. Recall that $y_i\geq 0$, and $x_i=I_i S_i$.

 \subsubsection{Volatility rates}
 To approximate the volatility rates $\sigma_1, \sigma_2,$ we propose the  standard deviation. And we need to compute the variances of  the distributions obtained in the  deterministic estimations in the previous sub sub section.
 \begin{eqnarray*}
 V_1= \frac{1}{n}\sum_{i=0}^{n-1}(\beta_i^*- \beta_{app})^2 \qquad  \mbox{and}\qquad
 V_2=  \frac{1}{n}\sum_{i=0}^{n-1}(\gamma_i^*- \gamma_{app})^2
 \end{eqnarray*}
 And first idea for   approximating   volatilities are:
 $$\sigma_1= \sqrt{V_1}, \qquad \sigma_2= \sqrt{V_2}.$$
But  in the numerical simulation section we have  to take into account the equilibrium condition $(\ref{bal})$ that  derives from the modeling.\\
A second idea  that could be better, is to take: 
 $$\sigma_i= \frac{\sqrt{V_1}+  \sqrt{V_2}}{2}, \quad i= 1, 2. $$ 
 \begin{rem}
\textbf{Before proceeding further, we would like to underline that the numerical tests that we shall realize below, are to be understood under hypotheses. If nothing is done on time, it could be possible to fall on  the  below predictions. And an invitation is to see how it is possible to organize minimal actions to  mitigate strongly  possible damage caused by the COVID 19 pandemic in a country like Senegal.}
 \end{rem}

\subsection{Simple SIR model}\label{simpleSIR}
First, we use the SIR model simulations for the Senegal case study. For the initial values, we use date already stated in section \ref{eda}. The first 3 days cases are given in Table \ref{senI0}. 
We have: $N=16,296,361$ until the year 2019 (see \cite{pyramid}), $I_0^*$ = Confirmed[0]=1, $S_0^* = N-I_0^*-R_0^*$. \\
To estimate parameters $\beta$ and $\gamma$, we call the procedure {\tt Optuma} package \cite{optuna} with python, an open source hyper-parameter optimization framework to automate hyper-parameter search.
For the contact rate ($\beta$) and the mean recovery rate ($\gamma$), we have: $\beta=0.115953$, $\gamma=0.030912$ (1/[days]) and $R_0=3.75$.  The prediction is given in Figure \ref{fig_sir}.
\par\vspace{-.5cm}
\begin{table}[h!] 
\begin{center}
\begin{tabular}{|c|c|c|c| } 
 \hline
          Date  &     Confirmed  &    Infected &    Deaths  \\
           \hline
2020-03-02  &	1 	&  1  &	0 \\
2020-03-04  &	2 	&  2  &	0 \\
2020-03-05  &	4 	&  4  &	0 \\
\hline
\end{tabular}
\end{center}
\caption{Senegal: first three days}\label{senI0}
\end{table}
\par\vspace{-1.cm}
\begin{center}
\begin{figure}[h!]
	\centering
	\includegraphics[width=.9\linewidth]{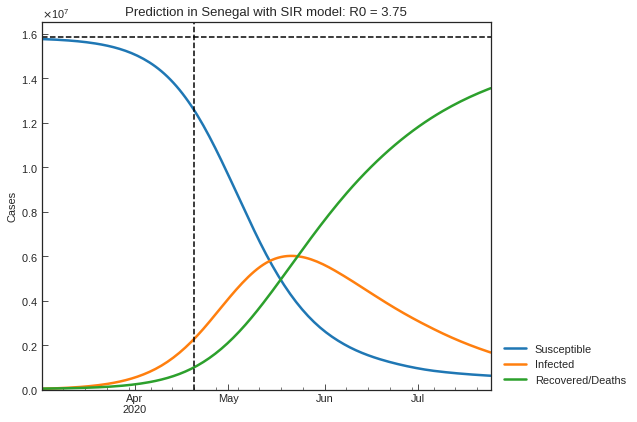}
	\caption{Senegal: prediction with SIR}\label{fig_sir}
\end{figure}
\end{center}
\par\vspace{-.5cm}

\subsection{Stochastic SIR model }
In this stochastic part, the following relation is to be considered
 \begin{eqnarray}
 -\sigma_1 dW_1+  \sigma_2 dW_2= 0
 \end{eqnarray}
 with  $dW = Z\sqrt{\Delta t}$, and  $Z \backsim \mathcal{N} (0 ,1 )$.\\
 For  the numerical simulation of the stopping time, we consider, $V$ a closed  subset of $\mathbb{R}^3$ in the form of $[0, S^*] \times [0, I^*]\times[0, R^*].$\\
We consider the same procedure like in section \ref{simpleSIR} with the same population and the same initialization, $I_0^*$ = Confirmed[0]=1, $S_0^* = N-I_0^*-R_0^*$. \\
To estimate parameters $\beta$ and $\gamma$, we call the procedure {\tt scipy.optimize.curve\_fit} with python.
For the contact rate ($\beta$) and the mean recovery rate ($\gamma$), we have: $\beta=0.135005$, $\gamma=0.026979$ (1/[days]), $\sigma_2=0.036254$ and $R_0=5$. As  $\sigma_1 dW_1 =  \sigma_2 dW_2$, we only use $\sigma_2 dW_2$ in the simulations. \\
Predictions with different simulations (because of Brownian) are given in Figure \ref{sn_ssir}. 
{\bf For the Brownian, the curve changes for each simulation}. For 4 tests, this gives us the results of Figures \ref{ssir_1}, \ref{ssir_2}, \ref{ssir_3} and \ref{ssir_4}.

\begin{rem}
The stopping time (see section \ref{stopping}) is illustrated in dotted line and appears around the middle April.
\end{rem}
\begin{figure}[h!]
  \subfloat[a first stochastic SIR]{
	\begin{minipage}[1\width]{0.5\textwidth}
	   \centering
	   \includegraphics[width=1.1\textwidth]{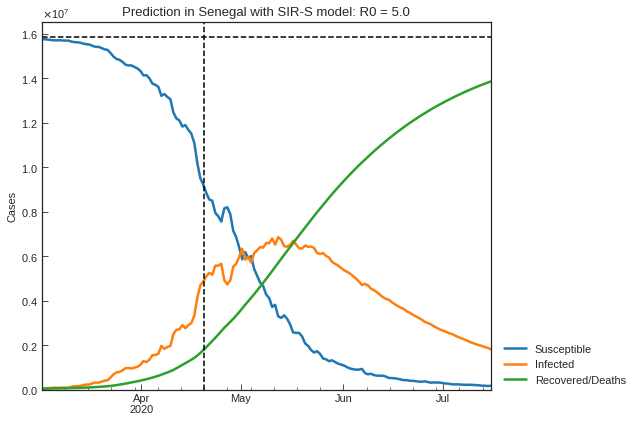}\label{ssir_1}
	\end{minipage}}
  \subfloat[a second stochastic SIR]{
	\begin{minipage}[1\width]{ 0.5\textwidth}
	   \centering
	   \includegraphics[width=1.1\textwidth]{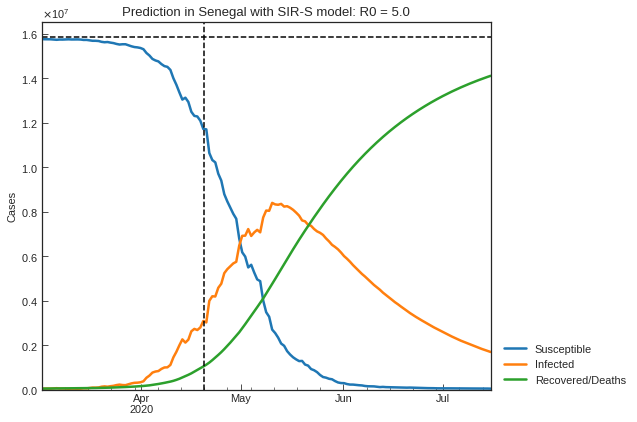}\label{ssir_2}
	\end{minipage}}
\newline
  \subfloat[a third stochastic SIR]{
	\begin{minipage}[1\width]{ 0.5\textwidth}
	   \centering
	   \includegraphics[width=1.1\textwidth]{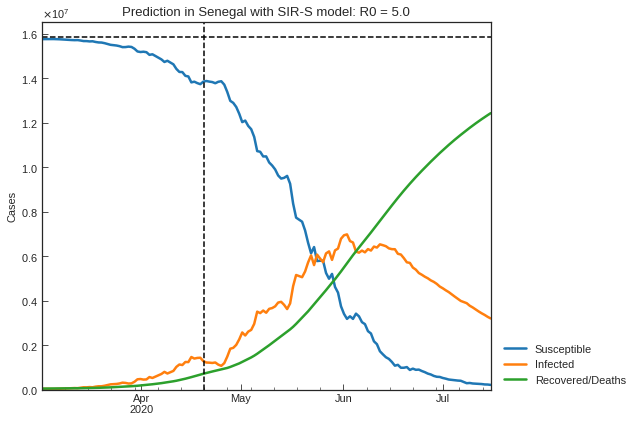}\label{ssir_3}
	\end{minipage}}
%
  \subfloat[a fourth stochastic SIR]{
	\begin{minipage}[1\width]{ 0.5\textwidth}
	   \centering
	   \includegraphics[width=1.1\textwidth]{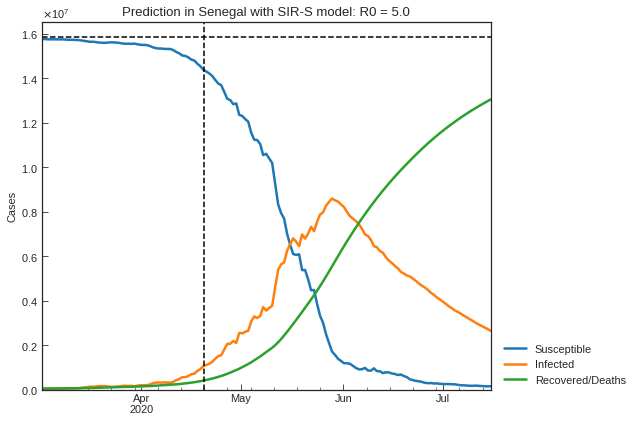}\label{ssir_4}
	\end{minipage}}
	
\caption{Senegal: prediction with stochastic SIR}\label{sn_ssir}
\end{figure}

\subsection{Interpretation of figures}
The values of the parameters having been estimated, the SIR deterministic model (Figure \ref{fig_sir}) shows that the peak of infection could be reached by mid-May with about 37.5\% of the population infected. The same time period of peak (mid-May) for infected, is observed for Figures \ref{fig_sird}, \ref{fig_sirf} and \ref{fig_sewir}. \\
\noindent The stochastic SIR model illustrates that it is possible to have nearly the same peak than the deterministic one (but with a larger infected population) if random factors are not too important (Figures \ref{ssir_1} and \ref{ssir_2}).\\
\noindent In the other hand if the effect of hazard is important, the stochastic SIR model is less optimistic and show that the peak of infection could be reached in early June (Figures \ref{ssir_3} and \ref{ssir_4}) with about 56\% of the population infected.\\
\noindent A finite stopping time exists (in dotted line) and is established in the second half of April.

\subsection{Others modifications of the SIR model}
There are a large number of modifications of the SIR model, including those that include births and deaths separately, where some cases are reported as fatal cases before clinical diagnosis of COVID-19, where the number of exposed cases in latent period and the waiting cases for confirmation are un-measurable variables, etc. 
All models allow for understanding how different situations may affect the outcome of the pandemic.

\subsubsection{SIR with Deaths \cite{diekmann,hethcote,Keeling} (SIR-D)}
It's possible to measure the number of fatal cases and recovered cases separately. We can use two variables Recovered and Deaths, instead of Recovered + Deaths in the mathematical model.\\
The model is given by:
\begin{equation}\label{model_sird}
\left\{ 
\begin{array}{ll}
\frac{dS}{dt}=-\beta IS\\[3mm]
\frac{dI}{dt}= \beta IS-(\gamma+\alpha) I\\[3mm]
\frac{dR}{dt}=\gamma I \\[3mm]
\frac{dD}{dt}=\alpha I
\end{array}\right.
\end{equation} 
We have:\\
$S$ : Susceptible, $I$ : Infected, $R$ : Recovered, $D$ : Fatal. \\
In addition, $S+I+R+D=N$, where $N$ is the total population, always obtained from \cite{pyramid}.
\noindent The basic reproduction number (also called basic reproduction ratio) is defined as  is $R_0 = \beta (\gamma+\alpha)^{-1}$. 
For the model, we also have:
$$ S+I \overset{\beta}{\longrightarrow} 2I, \ I \overset{\gamma}{\longrightarrow} R \ \  \mbox{and}\ \  I \overset{\alpha}{\longrightarrow} D  $$ 
with: $\beta$= effective contact rate [1/day], \ 
$\gamma$= recovery rate [1/day], $\alpha$= mortality rate [1/day]. 
We also use the optuna package to estimate the parameters. \\
We obtain: $\beta=0.002343$, $\gamma=0.042431$, $\alpha=0.008784$ (1/[days]) and $R_0=3.81$.  The prediction is given in Figure \ref{fig_sird}.
\begin{center}
\begin{figure}[h!]
	\centering
	\includegraphics[width=.9\linewidth]{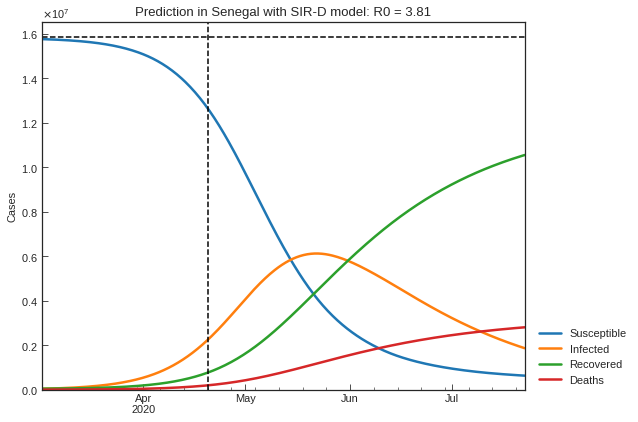}
	\caption{Senegal: prediction with SIR with Deaths}\label{fig_sird}
\end{figure}
\end{center}
\par\vspace{-.5cm}
\subsubsection{SIR with Fatal \cite{diekmann,hethcote,Keeling} (SIR-F)}
We can have a situation where some cases are reported as fatal cases before clinical diagnosis of COVID-19. To consider this issue, S + I $\longrightarrow$  
  Fatal + I will be added to the model.\\
The model is given by:
\begin{equation}\label{model_sirdf}
\left\{ 
\begin{array}{ll}
\frac{dS}{dt}=-\beta IS\\[3mm]
\frac{dI}{dt}= (1-\alpha_1) \beta IS-(\gamma+\alpha_2) I\\[3mm]
\frac{dR}{dt}=\gamma I \\[3mm]
\frac{dD}{dt}=\alpha_1 \beta IS + \alpha_2 I
\end{array}\right.
\end{equation} 
We have:\\
$S$ : Susceptible, $S^*$ : Confirmed and un-categorized, $I$ : Confirmed and categorized as $I$, $R$ : Recovered, $F$: Fatal with confirmation.\\ 
In addition, $S+I+R+F=N$, where $N$ is the total population, always obtained from \cite{pyramid}.
\noindent The basic reproduction number (also called basic reproduction ratio) is defined as  is $R_0 = \beta (1-\alpha_1)(\gamma+\alpha)^{-1}$. 
For the model, we also have:
$$ S\overset{\beta I}{\longrightarrow} S^* \overset{\alpha_1}{\longrightarrow} F,\  S\overset{\beta I}{\longrightarrow} S^* \overset{1-\alpha_1}{\longrightarrow} I,  \ I \overset{\gamma}{\longrightarrow} R \ \  \mbox{and}\ \  I \overset{\alpha_2}{\longrightarrow} F  $$ 
with: $\beta$= effective contact rate [1/day], \ 
$\gamma$= recovery rate [1/day], $\alpha_1$= mortality rate of $S^*$ cases [1/day], $\alpha_2$= mortality rate of $I$ cases [1/day].
We also use the Optuna package to estimate the parameters. \\
We obtain: $\beta=0.103144$, $\gamma=0.02308$, $\alpha_1=0.057568$, $\alpha_2=0.000192$ (1/[days]) and $R_0=4.18$.  The prediction is given in Figure \ref{fig_sirf}.
\begin{center}
\begin{figure}[h!]
	\centering
	\includegraphics[width=.9\linewidth]{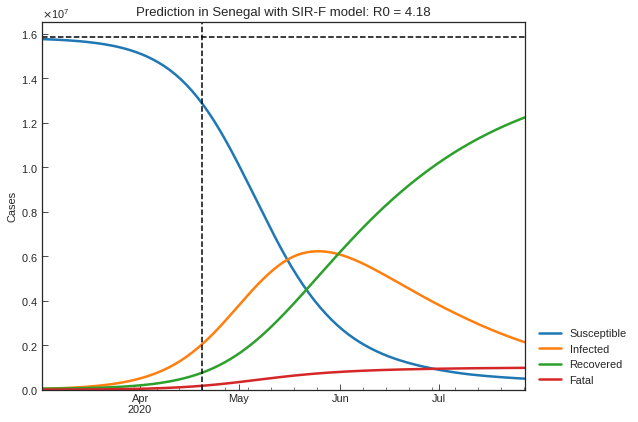}
	\caption{Senegal: prediction with SIR with Fatal}\label{fig_sirf}
\end{figure}
\end{center}
\par\vspace{-.5cm}

\subsubsection{SIR Exposed and Waiting cases with Fatal \cite{diekmann,hethcote,Keeling} (SEWIR-F)}
We can consider the number of exposed cases in latent period (E) and the waiting cases for confirmation (W) which are un-measurable variables. 
If E and W are large, outbreak will occur in the near future. W and some rules were added to explain COVID-19 dataset, but this is like-SEIR model. \\
The model is given by:
\begin{equation}\label{model_sirfew}
\left\{ 
\begin{array}{ll}
\frac{dS}{dt}=-\beta_1 (W+I)S\\[3mm]
\frac{dE}{dt}= \beta_1 (W+I)S-\beta_2 E\\[3mm]
\frac{dW}{dt}= \beta_2 E - \beta_3 W \\[3mm]
\frac{dI}{dt}= (1-\alpha_1) \beta_3 W-(\gamma+\alpha_2) I\\[3mm]
\frac{dR}{dt}=\gamma I \\[3mm]
\frac{dF}{dt}=\alpha_1 \beta_3 W + \alpha_2 I
\end{array}\right.
\end{equation} 
We have:\\
$S$ : Susceptible, $E$ : Exposed and in latent period (without infectivity), $W$ : Waiting cases for confirmation (with infectivity), 
$I$ : Confirmed and categorized as $I$, $R$ : Recovered and $F$: Fatal with confirmation.\\ 
In addition, Total population - Confirmed  = $S+E+W+S^*$, Confirmed = $I+R+F$, Recovered = $R$, Deaths = $F$ and  $S+E+W+I+R+F=N$, where $N$ is the total population, always obtained from \cite{pyramid}.
\noindent The basic reproduction number (also called basic reproduction ratio) is defined as  is $R_0 = \beta_1 (1-\alpha_1)(\gamma+\alpha_2)^{-1}$. 
For the model, we also have:
$$ S\overset{\beta_1 (W+I)}{\longrightarrow} E \overset{\beta_2}{\longrightarrow} W \overset{\beta_3}{\longrightarrow} S^* \overset{\alpha_1}{\longrightarrow} F,\  S  \overset{\beta_1 (W+I)}{\longrightarrow} E \overset{\beta_2}{\longrightarrow} W \overset{\beta_3}{\longrightarrow} S^* \overset{1-\alpha_1}{\longrightarrow} I  
$$
$$
\ I \overset{\gamma}{\longrightarrow} R \ \  \mbox{and}\ \  I \overset{\alpha_2}{\longrightarrow} F  $$ 
with: $\beta_1$= exposure rate (the number of encounter with the virus in a minute) [1/day], $\beta_2$= inverse of latent period [1/day], \ $\beta_3$= inverse of waiting time for confirmation [1/day], \ $\gamma$= recovery rate [1/day], $\alpha_1$= mortality rate of $S^*$ cases [1/day] ($S^*$ = Confirmed and un-categorized), $\alpha_2$= mortality rate of $I$ cases [1/day]. \\
We also use the Optuna package to estimate the parameters. \\ 
To estimate $\beta_2$ and $\beta_3$, we first calculate median value of latent period  $L_E$ and waiting time for confirmation $L_W$. We assume that  latent period is equal to incubation period (patients start to have infectivity from onset dates). \\
We obtain: $\beta_1=0.0.108008$, $\beta_2=0.407707$, $\beta_3=0.728274$, $\gamma=0.018478$, $\alpha_1=0.068890$, $\alpha_2=2.066063e-05$ (1/[days]) and $R_0=5.44$.  The prediction is given in Figure \ref{fig_sewir}.
\begin{center}
\par\vspace{-.5cm}
\begin{figure}[h!]
	\centering
	\includegraphics[width=.9\linewidth]{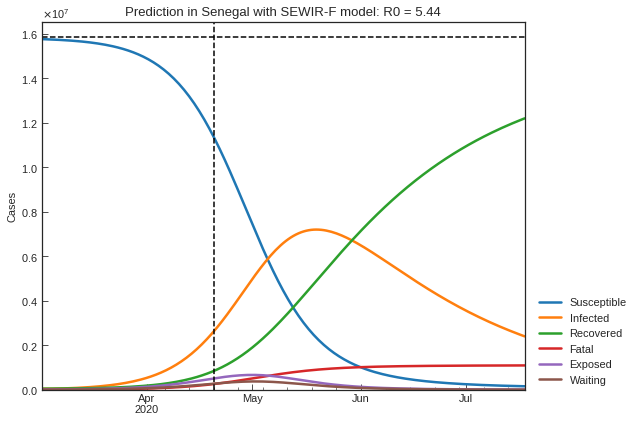}
	\caption{Senegal: prediction with SEWIR with Fatal}\label{fig_sewir}
\end{figure}
\end{center}
\par\vspace{-1.cm}
\begin{rem}
Note that it is quite possible to propose stochastic  variants for  all the above various  deterministic models.
\end{rem}

\subsection{Forecasting using Prophet}
In this section, we develop the machine learning technics for forecasting to compare with the previous SIR models. 
We use Prophet \cite{prophet, sean}, a procedure for forecasting time series data based on an additive model where non-linear trends are fit with yearly, weekly, and daily seasonality, plus holiday effects. \\
For the average method, the forecasts of all future values are equal to the average (or “mean”) of the historical data. If we let the historical data be denoted by $y_1,...,y_T$, then we can write the forecasts as
$$
\hat{y}_{T+h|T}=\bar{y}=(y_1+y_2+...+y_T)/T
$$
The notation $\hat{y}_{T+h|T}$ is a short-hand for the estimate of $y_{T+h}$  based on the data $y_1,...,y_T$. 
A prediction interval gives an interval within which we expect $y_t$  to lie with a specified probability. For example, assuming that the forecast errors are normally distributed, a 95\% prediction interval for the  $h$-step forecast is 
$$
\hat{y}_{T+h|T}\pm1.96\hat{\sigma_h}
$$
where  ${\sigma_h}$ is an estimate of the standard deviation of the $h$-step  forecast distribution.  \\
For the data preparation, when we are forecasting at country level, for small values, it is possible for forecasts to become negative. To counter this, we round negative values to zero. Also, no tweaking of seasonality-related parameters and additional regressors are performed.\\
\noindent We can carry out simulations for a longer time and forecast the potential trends of the COVID-19 pandemic. In Senegal, the predicted cumulative number of confirmed cases are first plotted for a shorter period of the next 7 days, and 3 weeks ahead forecast with Prophet, with 95\% prediction intervals. \\
The confirmed predictions for Senegal are given in Figures \ref{sn_1w} and \ref{sn_2w} (see Tables \ref{sn_1w_confcases} and \ref{sn_2w_confcases} for the value of the confidence interval).
\begin{table}[h!] 
\begin{center}
\begin{tabular}{|c|c|c|c| } 
 \hline
          ds   &         $ \hat{y}$ &    $\hat{y}_{lower}$ &    $\hat{y}_{upper}$ \\      
           \hline
2020-04-22 	& 390.082137 &	381.573661 	&398.870511 \\
\hline 
2020-04-23  & 	399.143136 &	389.013079  &	408.336247  \\ 
\hline
2020-04-24 &	411.338180 	 & 400.868963 	& 422.324577 \\
\hline
2020-04-25 &	420.807645  &	408.116029  &	432.783086\\
\hline
2020-04-26 &	432.420448 &	418.126379   &	446.991749  \\
\hline
\end{tabular}
\end{center}
\caption{Senegal: predicted cumulative confirmed cases 
$\sim$April 26, 2020.}\label{sn_1w_confcases}
\end{table} 
\par\vspace{-1.cm}
\begin{table}[h!] 
\begin{center}
\begin{tabular}{|c|c|c|c|} 
 \hline
      ds  &         $ \hat{y}$ &    $\hat{y}_{lower}$ &    $\hat{y}_{upper}$   \\
           \hline
2020-05-06 &	534.240980 &	483.962306  &	586.061842 \\  
\hline 
2020-05-07 &	543.301979 &	487.616917 	&	600.747825  \\  
\hline
2020-05-08 &	555.497023 &	498.895290 &	616.731785   \\   
\hline
2020-05-09 &	564.966488 &	501.869902 &	631.122912   \\   
\hline
2020-05-10 &	576.579291 &	509.887788 &	647.035185  \\
\hline
\end{tabular}
\end{center}
\caption{Senegal: predicted cumulative confirmed cases $\sim$May 10, 2020.}\label{sn_2w_confcases}
\end{table} 
\begin{figure}[h!]
	\centering
	\includegraphics[width=0.9\linewidth]{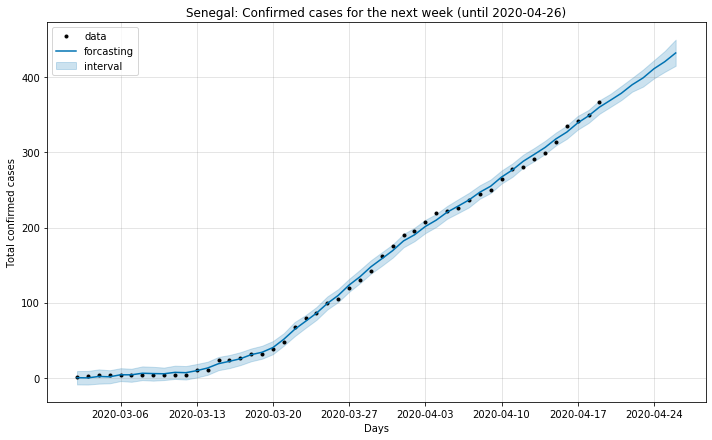}
	\caption{Senegal (Confirmed) : Forcasting for the next week $\sim$April 29, 2020 }  \label{sn_1w}
\end{figure}
\begin{figure}[h!]
	\centering
	\includegraphics[width=0.9\linewidth]{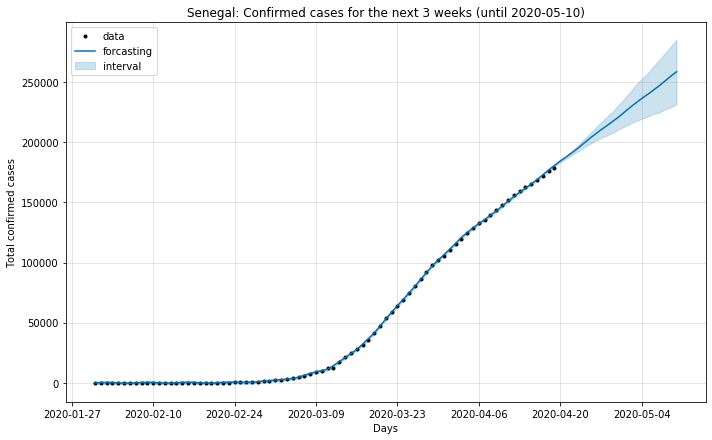}
	\caption{Senegal (Confirmed) : Forcasting for the next 3 weeks $\sim$May 10, 2020}   \label{sn_2w}
\end{figure}

\par\vspace{-.5cm}
\subsection{Main comments}
\noindent With Prophet, we perform also for the worldwide and three selected countries China, Italy, and Iran. Firstly, the worldwide predicted cumulative number of confirmed cases and deaths cases are plotted for a shorter period of the next 7 days. Secondly, we perform for 3 weeks for both confirmed and  deaths cases.
We plot only for the next three days for China, Italy, and Iran countries.\\
We can summarize our basic predictions as follows, for the worldwide and by country: 
\begin{itemize}
\item For Senegal (see Figure \ref{sn_comp}), the peak of the pandemic will be no later than the end of May. The predictions given by the SIR models and machine learning give roughly the same estimates. The authorities must take strict measures to stop the pandemic of COVID-19, because the peak can be reached around mid-May.
\item For worldwide, (see Figure \ref{ww_conf_deaths}), overall, each country of the whole world must take strict measures to stop the pandemic.  At $\sim$May 10, 2020 we may obtain $>$ 3 million 740 000 confirmed cases (see Table \ref{ww_2w_confcases}).

\item For Italy (see Figure \ref{it}) and Iran (see Figure \ref{ir}), the success of the anti-pandemic will be no later than the middle of May. The situation in Italy and Iran is still very severe. 
\item For China (see Figure \ref{cn}), based on optimistic estimation, the pandemic of COVID-19 would soon be ended within a few weeks in China.
\end{itemize}

\begin{figure}[h!]
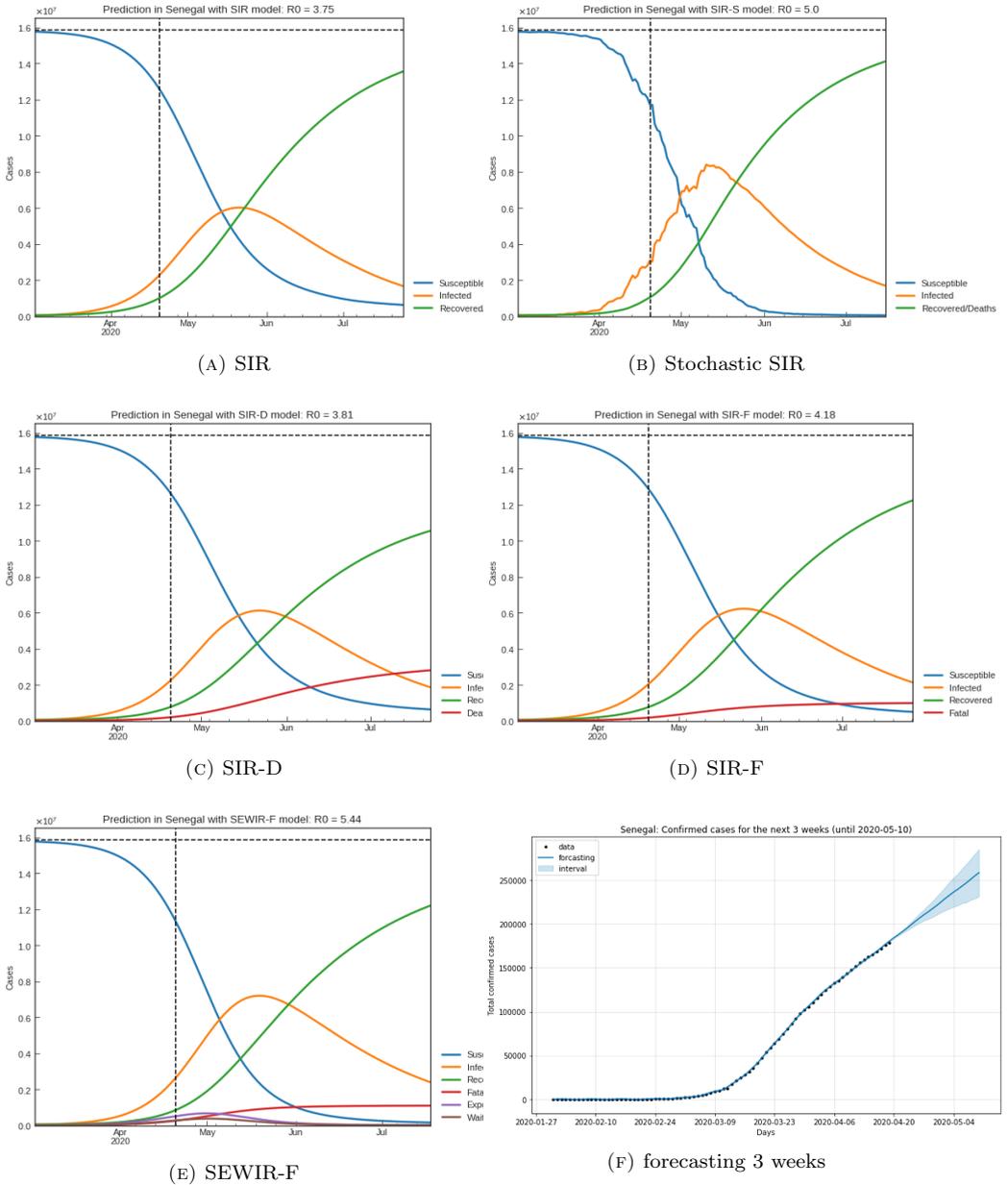

  \subfloat[SIR]{
	\begin{minipage}[1\width]{ 0.5\textwidth}
	   \centering
	   \includegraphics[width=1.1\textwidth]{figures/Senegal_Prediction_deterministe.png}
	\end{minipage}}
  \subfloat[Stochastic SIR]{
	\begin{minipage}[1\width]{ 0.5\textwidth}
	   \centering
	   \includegraphics[width=1.1\textwidth]{figures/ssir2.png}
	\end{minipage}}
	\newline
  \subfloat[SIR-D]{
	\begin{minipage}[1\width]{ 0.5\textwidth}
	   \centering
	   \includegraphics[width=1.1\textwidth]{figures/Senegal_Prediction_SIR-D.png}
	\end{minipage}}	
  \subfloat[SIR-F]{
	\begin{minipage}[1\width]{ 0.5\textwidth}
	   \centering
	   \includegraphics[width=1.1\textwidth]{figures/Senegal_Prediction_SIR-F.png}
	\end{minipage}}	
		\newline
  \subfloat[SEWIR-F]{
	\begin{minipage}[1\width]{ 0.5\textwidth}
	   \centering
	   \includegraphics[width=1.1\textwidth]{figures/Senegal_Prediction_SEWIR-F.png}
	\end{minipage}}	
  \subfloat[forecasting 3 weeks]{
	\begin{minipage}[1\width]{ 0.5\textwidth}
	   \centering
	   \includegraphics[width=1.1\textwidth]{figures/sn_3weeks_confirmed.png} 
	\end{minipage}}
		
\caption{Senegal: Comparative prediction between SIR models and machine learning technics}\label{sn_comp}
\end{figure}

\begin{table}[h!]
\begin{center}
\begin{tabular}{|c|c|c|c|} 
 \hline
      ds  &         $ \hat{y}$ &    $\hat{y}_{lower}$ &    $\hat{y}_{upper}$   \\
           \hline
2020-05-06  &	3.777760e+06 	&  3.512616e+06   &	4.006792e+06\\
2020-05-07  &	3.861820e+06 	&  3.573719e+06   &	4.106810e+06\\
2020-05-08  &	3.945427e+06 	&  3.633930e+06   &	4.213883e+06\\
2020-05-09  &	4.026816e+06 	&  3.691283e+06   &	4.311569e+06\\
2020-05-10   &	4.108123e+06 	&  3.740562e+06    &	4.415467e+06\\
\hline
\end{tabular}
\end{center}
\caption{Worldwide: predicted cumulative confirmed cases $\sim$May 10, 2020.}\label{ww_2w_confcases}
\end{table} 

\begin{figure}[h!]
  \subfloat[1 week confirmed]{
	\begin{minipage}[1\width]{ 0.5\textwidth}
	   \centering
	   \includegraphics[width=1.1\textwidth]{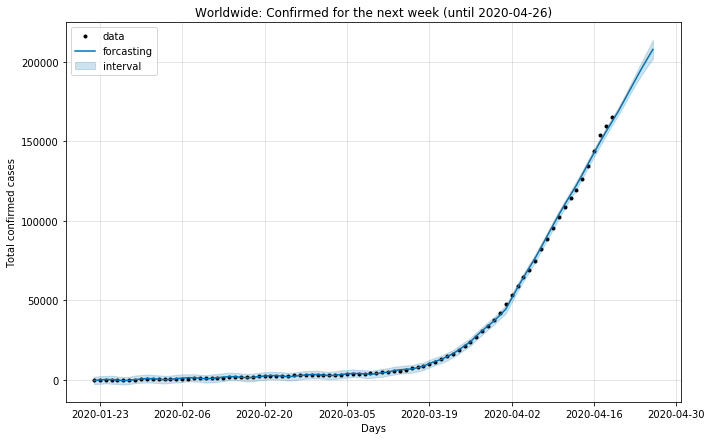}
	\end{minipage}}
  \subfloat[3 weeks confirmed]{
	\begin{minipage}[1\width]{ 0.5\textwidth}
	   \centering
	   \includegraphics[width=1.1\textwidth]{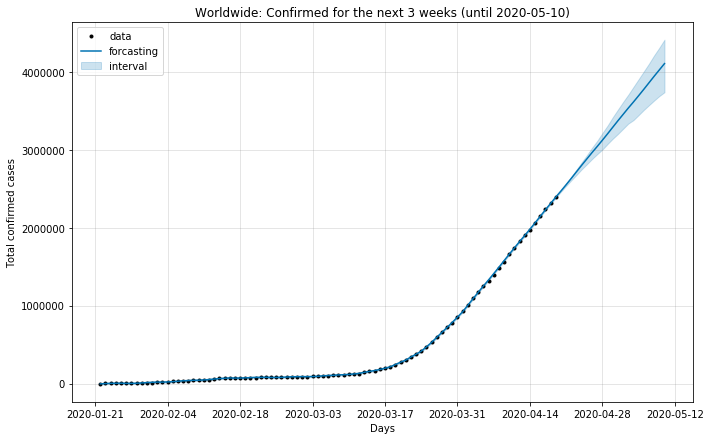}
	\end{minipage}}
	\newline
  \subfloat[1 week deaths]{
	\begin{minipage}[1\width]{ 0.5\textwidth}
	   \centering
	   \includegraphics[width=1.1\textwidth]{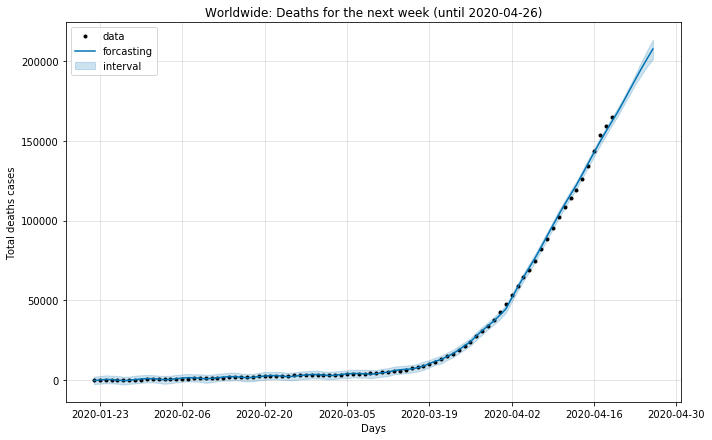}
	\end{minipage}}	
  \subfloat[3 weeks deaths]{
	\begin{minipage}[1\width]{ 0.5\textwidth}
	   \centering
	   \includegraphics[width=1.1\textwidth]{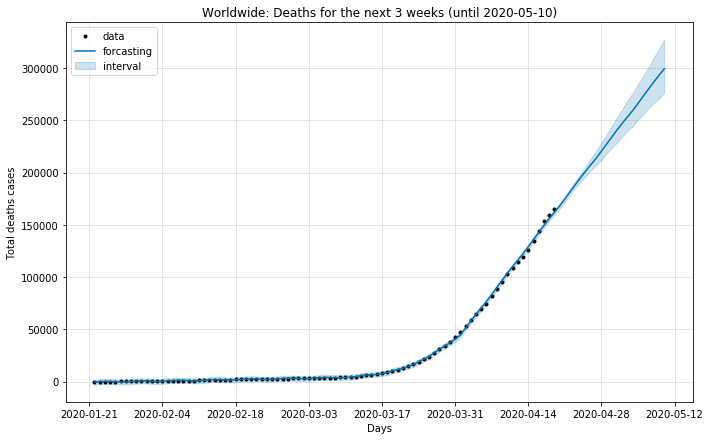}
	\end{minipage}}	
		
\caption{Worldwide: forecasting for confirmed \& deaths (next week and next 3 weeks)}\label{ww_conf_deaths}
\end{figure}

\begin{figure}[h!]
  \subfloat[Forcasting - confirmed in Italy]{
	\begin{minipage}[1\width]{ 0.5\textwidth}
	   \centering
	   \includegraphics[width=1.1\textwidth]{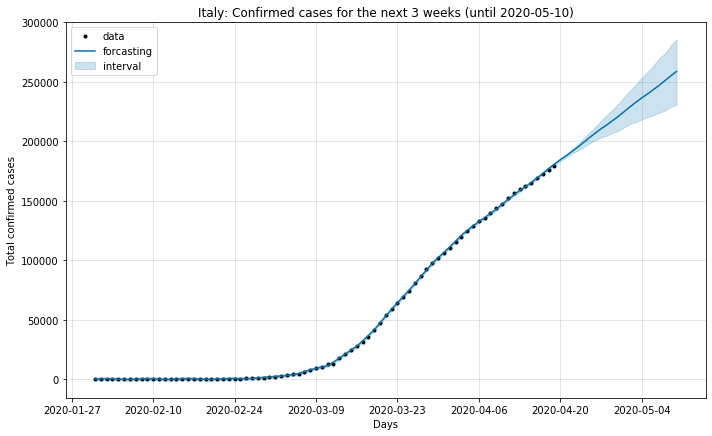}\label{it}
	\end{minipage}}
  \subfloat[Forcasting - confirmed in Iran]{
	\begin{minipage}[1\width]{ 0.5\textwidth}
	   \centering
	   \includegraphics[width=1.1\textwidth]{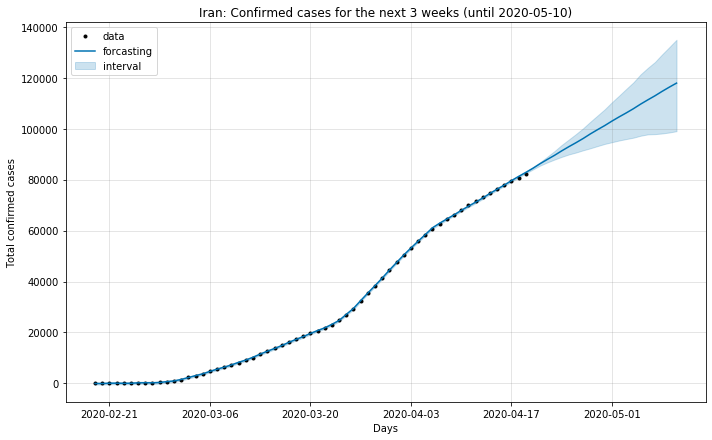}\label{ir}
	\end{minipage}}	
\newline
  \subfloat[Forcasting - confirmed in China]{
	\begin{minipage}[1\width]{ 0.5\textwidth}
	   \centering
	   \includegraphics[width=1.3\textwidth]{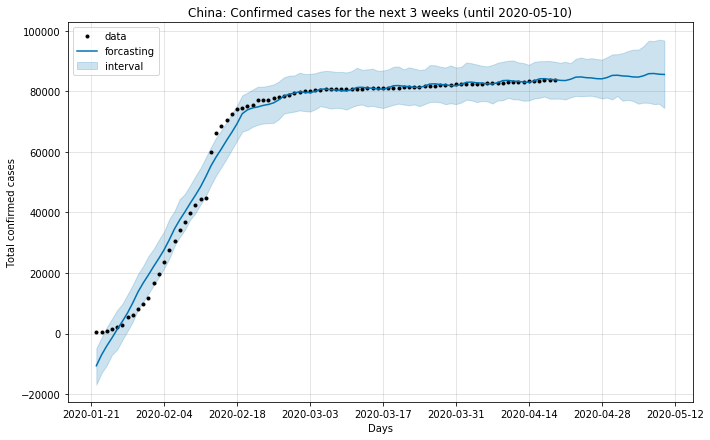}\label{cn}
	\end{minipage}}	
	
\caption{Italy, Iran, and China : forecasting for confirmed (next 3 weeks)}\label{ww_conf}
\end{figure}

\noindent Finally, due to the inclusion of suspected cases with clinical diagnosis into confirmed cases (quarantined cases), we can see severe situation in some cities, which requires much closer attention. 
Individuals, communities and governments have to fight against the spread of the coronavirus. And  thoughtful actions are to be taken.

\section{Conclusion and Perspectives}\label{ccl}
Under optimistic estimation, the pandemic in some countries (like China) will end soon within few weeks,  while for most of the countries in the world, the hit of anti-pandemic will be no later than mid-May. In Senegal, we expect the situation will end up at the beginning of May.\\
In front of good forecasting, it is  fundamental to get back to the estimation of parameters in general. And in the stochastic models, there is also another main issue: the estimation of  volatility parameters. These bring more works. And most of the time one considers the standard deviation to approximate them. But,  because of the difficulty to identify the asymptomatic cases,  like in Finance, should it be possible to introduce other ways to estimate them? One investigation way could be:
 \begin{itemize}
\item  $V= \displaystyle \frac{1}{n}\sum_{i=0}^{n-1} (x_i- \bar{x})^2 , \quad x_i, i= 0,...,  n-1$  are the measures given by the  infected numbers between two consecutive  periods:
 $x_i =  \ln (S_i, I_i, R_i, S_{i-1}, I_{i-1}, R_{i-1})$  (one could think to $x_i=  \ln( \frac {I_i}{I_{i-1}})$    provided that $, I_i, I_{i-1}>0$ with  possible other  conditions); $\bar{x}$ is statistical mean value of $x_i.$ And then $\sigma_1= \sqrt{V}.$\\
 \item If one thinks that the volatility cannot be equal to $0,$ i. e  we exclude the situation where $x_i= \bar{x}$ for any $i.$  
A possibility to take always into account the presence of volatility is  to use the approximation where the mean value of the sample is removed in the variance formula:  then one can consider the following estimation  for the 
 variance $V= \displaystyle \frac{1}{n}\sum_{i=0}^{n-1} x_i^2.$
 \end{itemize}
\noindent At the end of this work, we think that other questionings  could  merit to be studied such as
 \begin{itemize}
 \item  the deepening of the stochastic models, by considering the fractional Brownian motion, 
 \item the consideration  of  non-local terms in  some  models,
 \item some minimal actions   such as  stochastic optimization to control  the spreading of the disease,
 \item   and finally, the mean-field games could be an  interesting way  to investigate in that pandemic, lockdown problems.
 \end{itemize}

\subsection*{Acknowledgement}
The authors thanks the Non Linear Analysis, Geometry and Applications (NLAGA) Project for supporting this work. They thanks also the anonymous authors for their helpful comments.

\end{document}